\documentclass[11pt]{article}
\pdfoutput=1
\usepackage[margin=1in]{geometry}
\usepackage{amsmath,amssymb,latexsym,epsfig,amsthm}
\usepackage{authblk}
\usepackage{times}
\usepackage{color}
\usepackage{hyperref}
\usepackage{adjustbox}
\usepackage{enumerate}
\usepackage{appendix}
\usepackage[bottom]{footmisc}
\usepackage{algorithm}
\usepackage[noend]{algorithmic}
\usepackage[utf8]{inputenc}
\usepackage{float}
\usepackage{mathtools}
\DeclarePairedDelimiter{\ceil}{\lceil}{\rceil}

\newtheorem{theorem}{Theorem}
\newtheorem{lemma}{Lemma}

\newtheorem{obs}{Observation}
\newtheorem{cor}{Corollary}

\begin{document}





\title{A Variant of the Maximum Weight Independent Set Problem}
\author{Sayan Bandyapadhyay}
\affil{University of Iowa\\ Iowa City, USA}

\maketitle

\begin{abstract}
We study a natural extension of the Maximum Weight Independent Set Problem (MWIS), one of the most studied optimization problems in Graph algorithms. We are given a graph $G=(V,E)$, a weight function $w: V \rightarrow \mathbb{R^+}$, a budget function $b: V \rightarrow \mathbb{Z^+}$, and a positive integer $B$. The weight (resp. budget) of a subset of vertices is the sum of weights (resp. budgets) of the vertices in the subset. A $k$-budgeted independent set in $G$ is a subset of vertices, such that no pair of vertices in that subset are adjacent, and the budget of the subset is at most $k$. The goal is to find a $B$-budgeted independent set in $G$ such that its weight is maximum among all the $B$-budgeted independent sets in $G$. We refer to this problem as MWBIS. Being a generalization of MWIS, MWBIS also has several applications in Scheduling, Wireless networks and so on. Due to the hardness results implied from MWIS, we study the MWBIS problem in several special classes of graphs. We design exact algorithms for trees, forests, cycle graphs, and interval graphs. In unweighted case we design an approximation algorithm for $d+1$-claw free graphs whose approximation ratio ($d$) is competitive with the approximation ratio ($\frac{d}{2}$) of MWIS (unweighted). Furthermore, we extend Baker's technique \cite{Baker83} to get a PTAS for MWBIS in planar graphs. 
\end{abstract}

\section{Introduction}\label{sec:intro}
We are given an undirected graph $G=(V,E)$ and a \textit{weight} function $w: V \rightarrow \mathbb{R^+}$. An \textit{independent set} in $G$ is a subset of vertices such that no pair of vertices in that subset are connected with an edge. The \textit{weight} of a subset of vertices is the sum of the weights of the vertices in the subset. A \textit{Maximum Weight Independent Set} in $G$ is an independent set such that its weight is maximum among all the independent sets of $G$. We are given another function $b: V \rightarrow \mathbb{Z^+}$ which we refer to as \textit{budget} (or cost) function. The \textit{budget} of a subset of vertices is the sum of the budgets of the vertices in that subset. A \textit{$k$-budgeted independent set} in $G$ is an independent set such that its budget is at most $k$. The \textit{Maximum Weight Budgeted Independent Set Problem} or MWBIS is as follows.\\\\
\textit{Maximum Weight Budgeted Independent Set Problem} (MWBIS): Given a graph $G=(V,E)$, a weight function $w: V \rightarrow \mathbb{R^+}$, a budget function $b: V \rightarrow \mathbb{Z^+}$, and a positive integer $B$, find a $B$-budgeted independent set in $G$ such that its weight is maximum among all the $B$-budgeted independent sets in $G$.\\\\
A special version of MWBIS where weight of each vertex is 1, is referred to as \textit{Maximum Budgeted Independent Set Problem} or MBIS. Bar-Noy {\em et al.}~\cite{Bar-NoyBFNS01} mentioned about the MWBIS problem while considering problems on resource allocation and scheduling. But as far as we are concerned no work has been done till date on MBIS or on MWBIS.  

We note that the Maximum Weight Independent Set Problem (MWIS) where the goal is to find a maximum weight independent set, is a special case of MWBIS where the budget of the vertices can be considered to be uniform (or same) and $B$ is larger than the sum of the budgets of all vertices. MWIS is known to be $\mathcal{NP}$-hard even if the weight of any vertex is 1 (a version which is called Maximum Independent Set Problem or MIS). It is not possible to approximate MWIS within a factor of $|V|^{1-\epsilon}$ for any $\epsilon > 0$, unless NP=ZPP \cite{Hastad97}. Moreover, even if the maximum degree of the graph is at most 3, it is not possible to get a PTAS \cite{BermanF99}. As MWBIS generalizes MWIS all the hardness results for MWIS also hold for MWBIS. But being a generalization of MWIS, one of the most studied problems in the area of approximation and graph algorithms, it opens several directions for future research. For example, like in the case of MWIS one might be interested in studying this problem in special classes of graphs where it is possible to solve the problem exactly or to get a near optimum solution in polynomial time.

The primary motivation to study MWBIS comes from the following scenario. In case of MWIS one can select as many objects (vertices) as possible without violating the adjacency property. Also no selection cost is associated with any object. In other words the cost of selection of any object is same. But in many practical applications this cost varies from object to object and it is not possible to choose a set of objects whose total cost exceeds certain limit. MWIS has a plenty of applications in several fields of computer science including Scheduling, Wireless networks, Computer graphics, Map labelling, Molecular biology and so on. As MWBIS generalizes the MWIS problem it can be applied in most of the contexts where MWIS is applicable. Now we present two motivating applications of MWBIS.

\textit{Job Scheduling in a Computer}. In Job Scheduling Problem given a set of jobs that has to be executed in a computer, the goal is to find a maximum cardinality subset of jobs that can be executed without interfering with each other (\cite{Eva06}). Each job is specified by an interval (a starting time and a finishing time) during which the job is needed to run. A pair of jobs interfere with each other if the intervals corresponding to them intersect. Thus the Job Scheduling Problem is same as the MIS problem in an interval graph, which can be solved in polynomial time (\cite{GuptaLL82},\cite{Snoeyink07}). In practice the size of the jobs may vary and depending on the number of resources needed, the overhead associated with the jobs may vary a lot (\cite{Aida00},\cite{TakefusaMNAN99}). Moreover, each job has a priority associated with it, as the level of importance may be different for different jobs. Thus we consider the following variant of Job Scheduling Problem. Given a set of jobs, a weight (priority) and a resource requirement corresponding to each job, find a set of non-interfering jobs having maximum weight such that the total amount of resources allocated for running those jobs doesn't exceed the amount of resources available (\cite{Bar-NoyBFNS01}). 

\textit{Selecting Non-interfering Set of Transmitters}. In cellular mobile communication one of the most crucial problem is to reduce the level of interference (\cite{AndrewsCH07},\cite{suh08},\cite{viswa03}). Co-channel interference occurs when two transmitters use the same frequency. Thus usually the transmitters acting on same frequencies are kept a distance apart from each other so that signals from one transmitter doesn't affect the signals of the other. Now consider the following scenario. We are given a collection of transmitters that have profits and costs associated with it. The profit of a transmitter depends on its capability and the cost is mainly maintenance cost. A pair of transmitters interfere with each other if they use same frequency, and the signals they emit, overlap with each other. Moreover, the company interested in maintaining the transmitters has a certain budget. The goal is to find a collection of non-interfering transmitters which maximizes the profit so that the total cost of the transmitters doesn't exceed the budget. Chamaret {\em et al.}~\cite{cham97} consider a variant of this problem where two transmitters are allowed to be in the final solution if the ``amount'' of overlap of the corresponding signals is bounded by a certain threshold value. However, for simplicity they assume that the profits and the costs of the transmitters are uniform. 

Both of the above mentioned problems can be modelled using MWBIS. Thus considering the importance of MWBIS it is quite interesting to study this problem. Due to the hardness results implied from MWIS one might be interested in considering the problem in simple classes of graphs. The most convenient structure to consider is a tree. In any tree MWIS can be solved in linear time (\cite{tree2}). In case of interval graph if the vertices are sorted by the right endpoints of the corresponding intervals, MIS can be solved in linear time (\cite{GuptaLL82}). Polynomial time algorithms exist for MWIS also in other classes of graphs including bipartite graphs, line graphs, circle graphs, claw free graphs (having no induced $K_{1,3}$ as a subgraph), fork free graphs and so on (\cite{Alekseev04},\cite{Edmonds65},\cite{FaudreeFR97},\cite{Gavril73},\cite{Lozin06},\cite{Minty80},\cite{nakamura01},\cite{NashG10}). 
Baker (\cite{Baker83},\cite{Baker94}) designed an algorithm for MIS in any planar graph $G=(V,E)$, which ensures a $\frac{k}{k+1}$ factor approximation, and runs in $O(8^k|V|)$ time for fixed $k$.

Another class of graphs which has caught much attention in this context is the class of $d+1$-claw free graphs. A graph is called $d+1$-claw free if it doesn't contain $K_{1,d+1}$ as an induced subgraph. There are mainly two reasons for studying this class of graphs. Firstly, the intersection graphs of many geometric families of objects are $c$-claw free for some constant $c$. An intersection graph of a set of objects is constructed by considering a vertex for each object and an edge is drawn between two vertices if the corresponding objects intersect. The families of axis-parallel unit squares form 5-claw free graphs and families of unit circles form 7-claw free graphs. Secondly, $d+1$-claw free graphs are the largest class of graphs where the MWIS problem have constant ($d$) factor approximation ratio. There is a simple greedy algorithm which gives a $d$-factor approximation for MWIS in $d+1$-claw free graphs. But the challenge is to improve the approximation factor. Bafna {\em et al.}~\cite{BafnaNR96} gave a $\frac{d+1}{2}$-factor approximation for the MIS problem in any $d+1$-claw free graph. Their algorithm is based on local improvement technique. Hurkens and Schrijver \cite{Hurkens88} show that by increasing the size of allowed improvement it is possible to achieve an approximation factor close to $\frac{d}{2}$. Chandra and Halldorsson \cite{ChandraH01a} improved the approximation factor for MWIS to $\frac{2}{3}d$. Lastly, Berman \cite{Berman00} have designed an algorithm based on local improvement strategy which achieves an approximation factor of $\frac{d+1}{2}$ in any $d+1$-claw free graph.

The MWIS problem has also been studied in bounded degree graphs. Denote the maximum degree of any graph by $\Delta$. There is a simple greedy algorithm that achieves a $\Delta$ approximation ratio for MIS in any bounded degree graph. Halldorsson and Radhakrishnan \cite{HalldorssonR94} improved this approximation ratio to $\frac{\Delta}{6}+O(1)$. Vishwanathan \cite{Halldorsson98} proposed a SDP based $\frac{\Delta \log\log \Delta}{\log \Delta}$-factor approximation algorithm for the same problem. So far that is the best known approximation for MIS in any bounded degree graph. In the weighted case Halldorsson and Lau \cite{HalldorssonL97} gave an algorithm that achieves a $\frac{\Delta+2}{3}$ approximation ratio. Lastly, Halperin \cite{Halperin00} and Halldorsson \cite{Halldorsson00} independently design algorithms for MWIS whose approximation factor matches the best known approximation factor for MIS.

\subsection{Our Results and Techniques}
We have studied the MWBIS problem mainly in several special classes of graphs. We have designed exact algorithms in case of trees, forests, cycle graphs, and interval graphs. The time complexity of all of these algorithms is bounded by a polynomial in $n$ and $B$. In all of these cases the most challenging issue to address is how to distribute the budget $B$ among the vertices. That is where the MWBIS problem becomes harder compare to MWIS. For example, in case of trees the following recursive routine solves the MWIS problem. The routine runs on the root of the tree. There could be two cases (i) the root is in the solution (ii) the root is not in the solution. In the first case the routine is recursively called on the grandchildren of the root. The solution is composed of the root and the solutions returned by the grandchildren. In the second case the routine is recursively called on the children of the root and the solution is composed of the solutions returned by them. The larger solution is returned by the routine. Note that in case of MWBIS this simple routine doesn't work, as it has no idea how to distribute the budget among the children or grandchildren of the root. We address this issue in our work by using a routine based on Optimum Resource Allocation.

We have designed approximation algorithms in case of $d+1$-claw free graphs and bounded degree graphs. We show that a simple greedy algorithm gives a $d$-factor approximation for MBIS in any $d+1$-claw free graph. If the maximum degree of the graph is bounded by $\Delta$, then the same algorithm gives a $\Delta$-factor approximation for any arbitrary graph. 

Lastly, we extend Baker's technique \cite{Baker83} to design a PTAS for MWBIS in planar graphs. Though the approach is quite similar, in this case it is crucial to figure out how the budgets of the vertices can be handled while constructing the table for dynamic programming. We use a robust resource allocation routine to resolve this issue.
\\\\\textit{Organization of the paper}. In Section \ref{sec:tree} we design algorithms for trees, forests and cycle graphs. In Section \ref{sec:claw free-bounded degree} we discuss the approximation algorithm for $d+1$-claw free graphs. Section \ref{sec:planar} is devoted to the discussion on planar graphs. In Section \ref{sec:interval} we discuss the algorithm for interval graphs.

\section{MWBIS in Forests and Cycle Graphs}\label{sec:tree}
In this section we design algorithms to solve the MWBIS problem in forests and cycle graphs. First we design an algorithm to solve MWBIS in a given tree. Then we use this algorithm to solve MWBIS in any forest and cycle graph.

\subsection{MWBIS in Trees}
Given a rooted tree $T$ we consider the MWBIS problem in it. We design a dynamic programming based algorithm for this special version of MWBIS. Let $v_1,\ldots,v_n$ be the vertices of $T$. For better understanding at first we discuss a solution which partially works for this problem.

\subsubsection{A Partial Solution of MWBIS}
A good thing about the trees is that they have very simple structure. Due to this simplicity they possess some special characteristics which are not present in general graphs. For example consider the parent-child relationship between two nodes. We can get a partial ordering of the vertices based on this relationship. The following idea explores this partial ordering. For any $i$ consider the vertex $v_i$ with weight $w_i$ and budget $b_i$. Let $T_i$ be the subtree rooted at $v_i$. Now we describe a recursive subroutine which partially solves the MWBIS problem in $T_i$. Note that we are interested in finding a maximum $B$-budgeted independent set of $T_i$. Now there could be two cases, (i) $v_i$ is in the solution, and (ii) $v_i$ is not in the solution. In the first case no child of $v_i$ could be taken in the solution as it violates the property of independent set. Thus in this case the solution is composed of $v_i$ and the solutions returned by running the subroutine on the grandchildren of $v_i$. In the second case the solution is composed of the solutions returned by running the subroutine on the children of $v_i$. The larger solution among these two is returned as the solution for $T_i$. This subroutine works well if one is interested in computing just a Maximum Weight Independent Set. But in case of MWBIS we also have the budget constraint and thus we need a scheme to decide how the remaining budget should be distributed optimally among the children or grandchildren of $v_i$. We'll resolve this issue by using a resource allocation routine.

\subsubsection{A Solution Based on Resource Allocation}
The problem of allocation of budget among the vertices is similar to the \textit{Optimum Resource Allocation Problem} (\cite{ora1},\cite{ora2}). The \textit{Optimum Resource Allocation Problem} is as follows. We are given $p$ resources and $k$ processors. Corresponding to each processor $j$ there is an efficiency function $f_j$. $f_j(p_j)$ denotes the efficiency of processor $j$ when $p_j$ resources are allocated to it. The values $f_j(p_j)$ are also given for $0 \leq p_j\leq p$ and $1\leq j\leq k$. The goal is to find an allocation of the $p$ resources to $k$ processors so that $\sum_{j=1}^k f_j(p_j)$ is maximized, where $p_j$ resources are allocated to processor $j$ and $\sum_{j=1}^k p_j=p$. The following theorem is due to Karush \cite{karush}.

\begin{theorem}\label{th:alloc}
There is a routine ALLOC($f_1$, $\ldots$,$f_k$;$p$) which solves the Optimum Resource Allocation Problem in $O(kp^2)$ time.
\end{theorem}

To be precise ALLOC returns a $k$-vector $p_1,\ldots,p_k$ corresponding to the optimum resource allocation, and the optimum solution as well. We use ALLOC to solve MWBIS in $T$. The algorithm we design has two phases. In the first phase it computes the weight of maximum weight $j$-budgeted independent set in subtrees of $T$ that has some special property, where $j \leq B$. In the second phase the tree $T$ is traversed in top-down manner to retrieve the maximum weight $B$-budgeted independent set in $T$ computed in the first phase. Now before moving on we need the following definition. 

A subtree $T_i$ of $T$, rooted at $v_i$ is called a \textit{maximal subtree} if for any vertex $v'$ in $T_i$, $v$ is a child of $v'$ in $T$ implies $v$ is also present in $T_i$ as a child of $v'$. See Figure \ref{fig:fig1} for an illustration. We note that the maximal tree rooted at the root of $T$ is the tree $T$ itself.

\begin{figure*} 
 \centering
 \hspace{-1in}
  \begin{minipage}[c]{0.5\textwidth}
  \centering
  \includegraphics[width=40mm] 
    {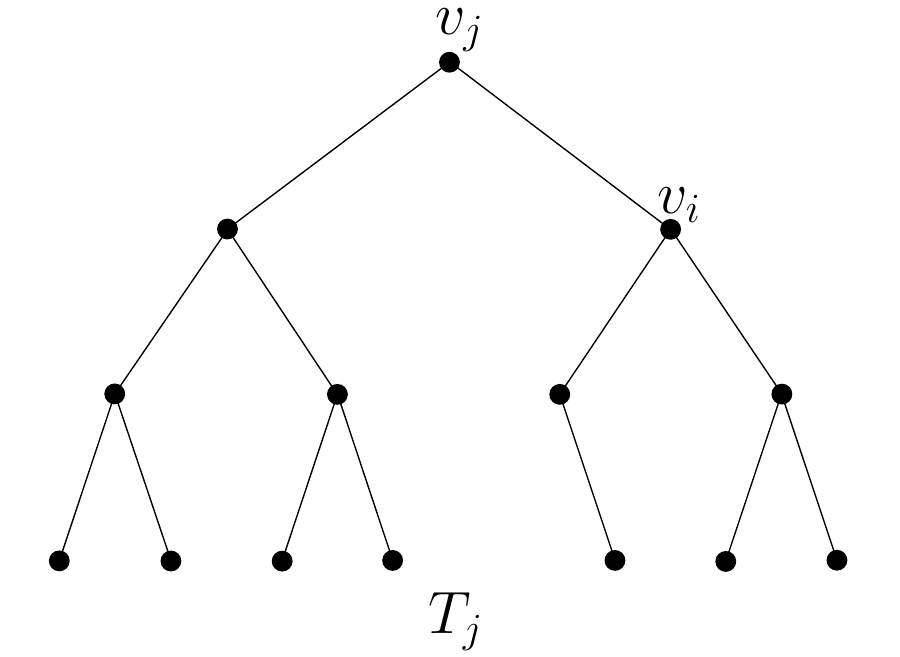}\\
    {\small (a)}\\
    \end{minipage}%
  \hspace{-1.6in}
  \begin{minipage}[c]{0.5\textwidth}
  \centering
  \includegraphics[width=42mm]
    {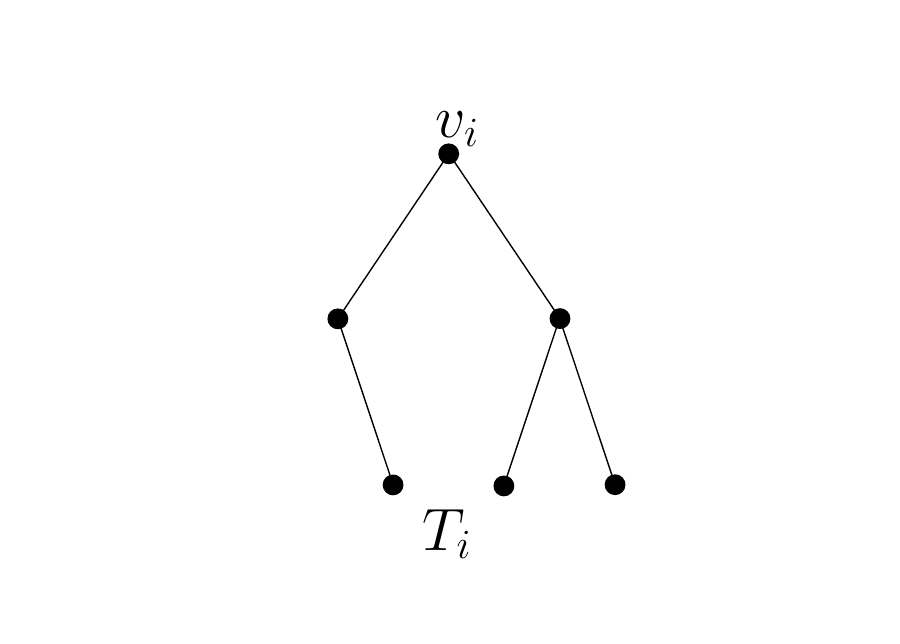}\\
    {\small (b)}\\
    \end{minipage}%
    \hspace{-1in}
    \hspace{-0.6in}
  \begin{minipage}[c]{0.5\textwidth}
  \centering
  \includegraphics[width=40mm]
    {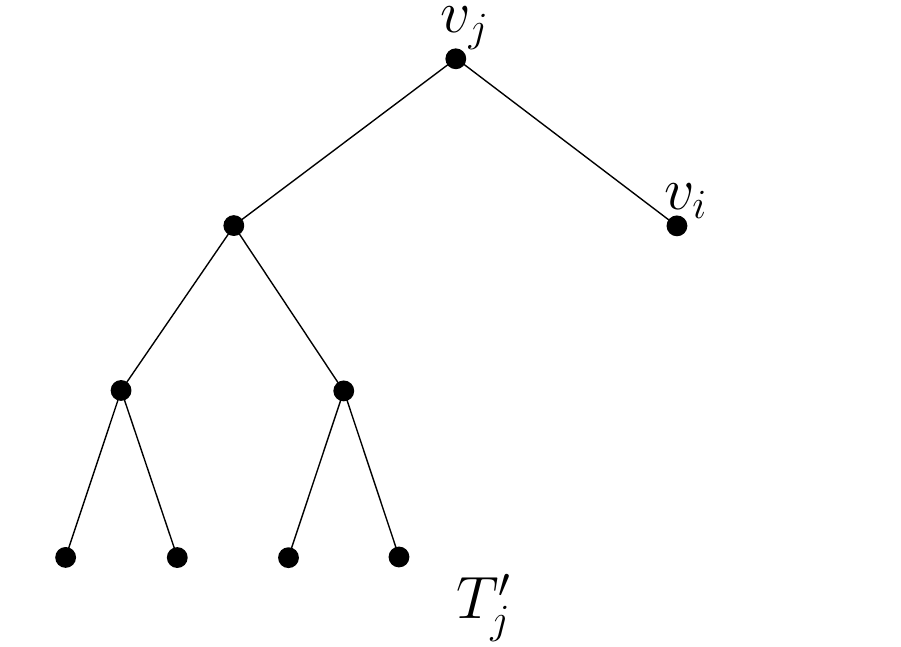}\\
    {\small (c)}\\
    \end{minipage}%
    \hspace{-1in}
 \caption{\textit{(a) A tree $T_j$ (b) the maximal subtree $T_i$ of $T_j$ rooted at $v_i$ (c) $T_j'$ is not a maximal subtree of $T_j$, as the children and grandchildren of $v_i$ are not present in $T_j'$}}
  \label{fig:fig1}
\end{figure*}

Like before consider any vertex $v_i$ with weight $w_i$ and budget $b_i$, and let $T_i$ be the maximal subtree rooted at $v_i$. In the first phase we design a routine OPT$(i,j)$ which computes the weight of maximum weight $j$-budgeted independent set in $T_i$. First we define the routine in an informal manner. Let $C_i$ and $G_i$ be the respective set of indexes of the children and grandchildren of $v_i$. As mentioned before there is two cases. If $v_i$ is selected in the solution, $b_i$ budget has already been used. The remaining budget $j-b_i$ should be distributed among the vertices contained in the maximal subtrees rooted at the grandchildren of $v_i$. This is similar to the problem of allocating $j-b_i$ resources to $|G_i|$ processors. Thus we call ALLOC with the function $f_t(p_t)$=OPT$(i_t,p_t)$, where $i_t \in G_i$ for $1\leq t\leq |G_i|$ and $p=j-b_i$. In this case the solution is the sum of $w_i$ and the optimum value computed by ALLOC. Note that for this case to be true $j$ should be greater than or equal to $b_i$. In the second case as $v_i$ is not selected ALLOC is called with $f_t(p_t)$=OPT$(i_t',p_t)$, where $i_t' \in C_i$ for $1\leq t\leq |C_i|$ and $p=j$. In this case the solution is just the optimum value computed by ALLOC. Let $s_1$ and $s_2$ be the solutions in Case (i) and Case (ii) respectively. OPT$(i,j)$ returns the maximum of $s_1$ and $s_2$. 

Note that before making a call to ALLOC we need to ensure that the values of the function $f_t$ are already computed for all $t$. Thus any top-down recursive approach does not work in this case. Instead we adapt a bottom-up dynamic programming based approach. A table $D$ is used to store the values computed by OPT. Each row of $D$ is corresponding to a maximal subtree $T_i$ and each column is corresponding to a budget $j$, where $1 \leq j\leq B$. The cell $D(i,j)$ stores two values: $w_{i,j}$, the weight of maximum weight $j$-budgeted independent set in $T_i$ and $f_{i,j}$, a flag that indicates if $v_i$ is in the solution ($f_{i,j}=1$) or not ($f_{i,j}=0$). Moreover, $D(i,j)$ stores the distribution of the budget $j$ among $v_i$ and the subtrees rooted at the children or grandchildren of $v_i$, as a vector. Thus each cell $D(i,j)$ can store at most $n+1$ values. Now we formally define the subroutine OPT$(i,j)$.

\begin{description}

\item[OPT$(i, j)$:] If $v_i$ is a leaf of $T$ and if $j \geq b_i$, then assign $w_i$ to $w_{i,j}$, 1 to $f_{i,j}$ and terminate.

\item If $v_i$ is a leaf of $T$, but $j < b_i$, then assign $0$ to both $w_{i,j}$ and $f_{i,j}$, and terminate.

\item If $j < b_i$, then assign 0 to $s_1$. Otherwise, assign the maximum solution computed by ALLOC$(f_1, \ldots, f_{|G_i|};$ $p)$ to $s_1$, where $f_t(p_t)=w_{i_t,p_t}$, $i_t \in G_i$ for $1\leq t\leq |G_i|$ and $p=j-b_i$. 

\item Assign the maximum solution computed by ALLOC$(f_1,\ldots,f_{|C_i|};p)$ to $s_2$, where $f_t(p_t)$= $w_{i_t',p_t}$, $i_t' \in C_i$ for $1\leq t\leq |C_i|$ and $p=j$. 

\item If $s_1 + w_i > s_2$, then assign $s_1+w_i$ to $w_{i,j}$, 1 to $f_{i,j}$ and store the $|G_i|$-vector returned by the first call to ALLOC in $D(i,j)$.

\item Otherwise, assign $s_2$ to $w_{i,j}$, 0 to $f_{i,j}$ and store the $|C_i|$-vector returned by the second call to ALLOC in $D(i,j)$.
\end{description}

Note that here to compute the value of $w_{i,j}$ for a subtree, the values corresponding to its children and grandchildren should be computed before. To ensure this we consider the ordering of the vertices corresponding to the postorder traversal of $T$. Let $v_1,\ldots,v_n$ be the postorder traversal of the vertices. The entries of $D$ are computed in this order. To be precise suppose $v_l$ appears before $v_m$ in this ordering, then all the entries of $D$ corresponding to $T_l$ are computed before computation of the entries corresponding to $T_m$. Moreover, before computing the values corresponding to $D(i,j)$ all the values corresponding to $D(i,j')$ are computed before for all $j' < j$. Thus while calling the routine OPT$(i,j)$ all the required values of $D$ are already computed and OPT$(i,j)$ returns the optimum solution of MWBIS in $T_i$. The algorithm is as follows.

\begin{algorithm*}[hbt]
  \caption{}
 \label{alg:tree}
 \begin{algorithmic}[1]
    \REQUIRE A tree $T$, postorder traversal $v_1,\ldots,v_n$ of the vertices of $T$, budget function $b$, a positive integer $B$, a $n\times (B+1) \times (n+1)$ table $D$ 
    \ENSURE a maximum weight $B$-budgeted independent set in $T$ and its weight\\\COMMENT{The First Phase}
    \FOR {$i=1$ to $n$}    
      \FOR {$j=0$ to $B$}
	    \STATE call OPT($i$,$j$)
      \ENDFOR	  
    \ENDFOR
    \COMMENT{The Second Phase}
    \STATE $I \leftarrow \phi$        
    \STATE $I \leftarrow$ ComputeBudgetedIndSet($I$, $n$, $B$)
    \RETURN $I, w_{n,B}$    
 \end{algorithmic}
\end{algorithm*}

\begin{algorithm*}[hbt]
  \caption{: \textbf{Procedure} ComputeBudgetedIndSet($I$, $i$, $j$)}       
 \label{alg:computeindset}
 \begin{algorithmic}[1]        
    \IF {$j==0$}
      \RETURN $I$
    \ENDIF
    \IF {$f_{i,j}==1$}
      \STATE $I \leftarrow I\cup {v_i}$
      \STATE $j \leftarrow j-b_i$
      \FOR {\textbf{each} grandchild $v_{i_t}$ of $v_i$}
        \STATE $I \leftarrow$ ComputeBudgetedIndSet($I$, $i_t$, $j_t$) \COMMENT{$j_t$ is the budget allocated to the subtree rooted at $v_{i_t}$, which is stored in $D(i,j)$}     
      \ENDFOR
    \ELSE     
      \FOR {\textbf{each} child $v_{i_t'}$ of $v_i$}
        \STATE $I \leftarrow$ ComputeBudgetedIndSet($I$, $i_t'$, $j_t'$) \COMMENT{$j_t'$ is the budget allocated to the subtree rooted at $v_{i_t'}$, which is stored in $D(i,j)$}     
      \ENDFOR        
    \ENDIF
    \RETURN $I$
 \end{algorithmic}    
\end{algorithm*}

The sole job in the second phase is to retrieve the maximum weight $B$-budgeted independent set of $T$ computed in the first phase. For that purpose we design the ComputeBudgetedIndSet procedure. Compute BudgetedIndSet($I$, $i$, $j$) traverses the tree rooted at $v_i$ recursively, starting with $v_i$ and returns the maximum weight $j$-budgeted independent set $I$ of $T_i$ computed in first phase. 
While traversing, the procedure consults the entries of $D$. If the flag in $D(i,j)$ is on, then $v_i$ was selected in the solution and thus it is added to $I$. The procedure then traverses the subtrees rooted at grandchildren of $v_i$ recursively. Note that the budget allocated to each of those subtrees is already computed and stored in $D(i,j)$. Otherwise, $v_i$ was not selected in the solution and the procedure recursively traverses the subtrees rooted at children of $v_i$. In this case also the budget allocated to each of these subtrees is already computed and stored in $D(i,j)$. The recursion bottoms out if the procedure is called on a tree where zero budget is allocated ($j=0$). As the vertices are processed in the postorder traversal order, $v_n$ must be the root of $T$. Thus a call to ComputeBudgetedIndSet($I$, $n$, $B$) returns the maximum $B$-budgeted independent set of $T$. The entry $w_{n,B}$ stores the weight of this set.

Now we discuss the time complexity of our algorithm. In the first phase Algorithm \ref{alg:tree} makes $O(nB)$ calls to OPT. The time complexity of OPT is dominated by the two calls to ALLOC. As the number of children and grandchildren can be at most $n-1$, by Theorem \ref{th:alloc} a call to ALLOC takes $O(nB^2)$ time. Thus the time complexity of the first phase is $O(n^2B^3)$. In the second phase Algorithm \ref{alg:tree} makes a single call to the recursive procedure ComputeBudgetedIndSet. As ComputeBudgetedIndSet scans each vertex at most once it runs in linear time. Thus the time complexity of this algorithm is dominated by the complexity of the first phase. Hence we obtain the following theorem.

\begin{theorem}\label{th:on tree}
The MWBIS problem can be solved in any tree in $O(n^2B^3)$ time, where $n$ is the number of vertices of the tree.
\end{theorem}

\subsection{MWBIS in Forests}
We are given a forest $F=\{H_1,\ldots,H_k\}$, where each $H_i$ is a tree. Let $n$ and $n_i$ be the number of vertices of $F$ and $H_i$ respectively for $1\leq i\leq k$. Note that to solve MWBIS in $F$ the main issue is to decide how to distribute the budget $B$ among the trees in $F$. Thus here also we can use the ALLOC routine to resolve this issue. First we solve MWBIS in $H_i$ using Algorithm \ref{alg:tree} with budget $B$, for each $1 \leq i\leq k$. Then we make a call to ALLOC with $f_t(p_t)=OPT(H_t,p_t)$ for $1\leq t\leq k$ and $p=B$. ALLOC returns the maximum solution of MWBIS in $F$. The maximum weight $B$-budgeted independent set of the forest can be retrieved by traversing each tree in it using the ComputeBudgetedIndSet procedure.

To solve MWBIS in $H_i$ $O(n_i^2B^3)$ time is needed. Thus in total for all $1\leq i\leq k$ we need $O(B^3\sum_{i=1}^k n_i^2)$ = $O(n^2B^3)$ time. The additional call to ALLOC require $O(kB^2)=O(nB^2)$ time. Thus we have the following corollary of Theorem \ref{th:on tree}.

\begin{cor}\label{cor:forest}
The MWBIS problem can be solved in any forest in $O(n^2B^3)$ time, where $n$ is the number of vertices of the forest.
\end{cor}

\subsection{MWBIS in Cycle Graphs}
We consider the MWBIS problem in a given cycle graph $C_n$ with $n$ vertices. We note that a cycle graph is a closed path. We show that the routine which solves the MWBIS problem in any simple path can be called multiple times to solve the problem in any cycle graph. The following is a crucial observation for our solution and follows from the definition of independent set.

\begin{obs}\label{obs:exclude vertex}
Any independent set in a connected graph excludes at least one vertex of the graph.
\end{obs}

Applying Observation \ref{obs:exclude vertex} on any cycle graph yields that any maximum independent set excludes at least one vertex. We note that removing any vertex from a cycle graph gives a simple path. The algorithm for solving MWBIS in $C_n$ is as follows. For each vertex $v$ in $C_n$ we remove $v$ from $C_n$ and solve the problem in the corresponding simple path with $n-1$ vertices using Algorithm \ref{alg:tree}. The maximum among these $n$ solutions yields the optimum solution for $C_n$. Thus from Theorem \ref{th:on tree} we have the following corollary.

\begin{cor}\label{cor:cycle}
MWBIS can be solved in any cycle graph in $O(n^3B^3)$ time, where $n$ is the number of vertices of the cycle graph.
\end{cor}

We note that considering the simplicity in the structure of path the time complexity to solve MWBIS in any path can be reduced using a better algorithm. This in turn reduces the time complexity of MWBIS in any cycle graph.

\section{MBIS in d+1-claw free graphs}\label{sec:claw free-bounded degree}
%
We consider the MBIS problem in a given $d+1$-claw free graph $G$. We design a $d$-factor approximation algorithm for this problem. Consider the following simple greedy algorithm at first. Fix any ordering of the vertices. Process the vertices in this order in the following way. Select a vertex in the solution if none of its neighbor is already selected. Otherwise, skip it and process the next vertex. This simple algorithm gives a $d$-factor approximation for the Maximum Independent Set (MIS) Problem. Consider the solution $S$ given by this algorithm and any maximum solution $O$ of MIS. We delete all the vertices from $S$ and $O$ which are in $S\cap O$. We construct a bipartite graph with the set of vertices $S\cup O$. There is an edge between a vertex of $S$ to a vertex of $O$ if they are neighbors in $G$. As $G$ is $d+1$-claw free there can be at most $d$ neighbors (independent) of a vertex of $S$ in this bipartite graph. Also note that there is no isolated vertex in this bipartite graph which is in $O$. If not, then our algorithm could return a solution of size one more by adding it to the solution, as none of its neighbor is selected. Thus cardinality of $O$ can be at most $d|S|$. Thus our greedy algorithm gives a $d$-factor approximation for MIS in $G$. 

Now if we consider a maximum solution $O'$ of MBIS and construct a bipartite graph like before, there could be an isolated vertex in this graph which is in $O'$. This is because it could be the case that though none of its neighbor is selected in $S$, it was not chosen due to unavailability of budget. As the number of such isolated vertices could be large depending on the instance, we cannot ensure a $d$-factor approximation using this algorithm. In the next subsection we modify this greedy algorithm to get a $d$-factor approximation. The modified algorithm is called Minimum Budget First or MBF.

\subsection{Minimum Budget First: The Algorithm}
We have made a simple modification to our previous greedy algorithm to design MBF. Instead of processing the vertices in any arbitrary order the vertices are ordered in non-decreasing order of their budgets. And, then they are processed like before. Let $v_1,\ldots,v_n$ be the set of vertices in non-decreasing order of their budgets and $b_1,\ldots,b_n$ be their corresponding budgets. Now we formally describe the algorithm.

\begin{algorithm*}[hbt]
 \caption{Minimum Budget First}
\label{alg:mbf}
\begin{algorithmic}[1]
   \REQUIRE A $d+1$-claw free graph $G$, set $\{v_1,\ldots,v_n\}$ of vertices of $G$ ordered in non-decreasing order of their budgets, budget $b_i$ of $v_i$ for $1\leq i\leq n$, a positive integer $B$
   \ENSURE a $B$-budgeted independent set in $G$
   \STATE $A_1 \leftarrow v_1$
   \STATE $B_2\leftarrow B-b_1$
   \FOR {$i=2$ to $n$}
     \IF {$B_i < b_i$}
       \RETURN $A_{i-1}$ 
     \ENDIF
     \IF {$A_{i-1} \cup \{v_i\}$ forms an independent set}
       \STATE $A_i \leftarrow A_{i-1} \cup \{v_i\}$
       \STATE $B_{i+1} \leftarrow B_{i}-b_i$
     \ELSE
       \STATE $A_i \leftarrow A_{i-1}$
       \STATE $B_{i+1} \leftarrow B_{i}$
     \ENDIF  	  
   \ENDFOR
   \RETURN $A_n$
\end{algorithmic}
\end{algorithm*}

MBF processes the vertices in the following way. Before iteration $i$ let $A_{i-1}$ be the set of vertices already chosen to be in the solution and $B_i$ be the remaining budget which can be used. If the remaining budget $B_i$ is less than the budget of the current vertex the algorithm returns $A_{i-1}$ as the solution. Otherwise, if $v_i$ doesn't have any neighbor in $A_{i-1}$, $A_{i-1} \cup \{v_i\}$ is assigned to $A_i$ and $B_{i}-b_i$ is assigned to $B_{i+1}$. In the remaining case $A_{i-1} $ is assigned to $A_i$ and $B_{i}$ is assigned to $B_{i+1}$. The algorithm terminates if the condition $B_i < b_i$ is true or all the vertices are processed. Now we move on towards the analysis of this algorithm.

\subsection{The Analysis}
We prove that this simple greedy algorithm returns a solution which is at least a factor of $\frac{1}{d}$ of any optimum solution. We use a charging argument to bound the cardinality of any optimum set. The charging argument ensures that at most $d$ vertices in an optimum set are charged against a vertex in the solution returned by MBF. This in turn shows that the cardinality of the optimum set can be at most $d$ times of the cardinality of the solution returned by MBF. 

Suppose $A$ be the independent set returned by MBF and $A^*$ be an optimum solution. Without loss of generality we assume $A \cap A^*=\phi$, otherwise, we could charge any common vertex in one to one manner. Let $A=\{v_{i_1},v_{i_2},\ldots,v_{i_k}\}$ such that $b_{i_1}\leq b_{i_2}\leq \ldots \leq b_{i_k}$. 
For $1 \leq j\leq k$ we define the sets $Y_j, X_j, A_j^*$ in an inductive way. In base case consider $v_{i_1}$. Let $Y_1=A^*$ and $X_1$ be the vertices in $Y_1$ adjacent to $v_{i_1}$. Note that $Y_1$ is an independent set and as $G$ is a $d+1$-claw free graph there could be at most $d$ vertices in $A^*$ adjacent to $v_{i_1}$. Thus $|X_1| \leq d$. If $X_1$ is non-empty, set $A_1^*$ to be $X_1$. Otherwise, let $m_1=\arg \min_{v_t \in Y_1} b_t$. Set $A_1^*$ to be $\{m_1\}$. Thus in either case $|A_1^*| \leq d$. Now let $c_1 = \min_{v_t \in A_1^*} b_t$. As $v_{i_1}$ is the first vertex chosen by MBF it has minimum budget and thus $b_{i_1} \leq c_1$. Now in induction case let $A_t^*$ is defined for all $t < j$. Let $Y_j=Y_{j-1}\setminus A_{j-1}^*$. If $Y_j$ is $\phi$, set $X_j$ and $A_j^*$ to $\phi$. Otherwise, do the following. Let $X_j$ be the vertices in $Y_j$ adjacent to $v_{i_j}$. Thus $|X_j| \leq d$. If $X_j$ is non-empty, set $A_j^*$ to be $X_j$. Otherwise, let $m_j=\arg \min_{v_t \in Y_j} b_t$. Set $A_j^*$ to be $\{m_j\}$. Thus in either case $|A_j^*| \leq d$. Now let $c_j = \min_{v_t \in A_j^*} b_t$. 

Later we will prove that $A^* \setminus \cup_{j=1}^k A_j^*=\phi$. As $|A_j^*|\leq d$ for $1\leq j\leq k$ this implies $|A^*|\leq kd$ and thus MBF is a $d$-factor approximation algorithm. 

Note that if $Y_j$ is $\phi$, $c_j$ is not defined. Consider $j \geq 1$ such that $c_j$ is defined. The following lemma proves that $b_{i_j} \leq c_j$. 

\begin{lemma}\label{lem:lessbudget}
For any $j \geq 1$ if $Y_j$ is non-empty, then $b_{i_j} \leq c_j$. 
\end{lemma}

\begin{proof}
We prove this lemma using strong induction on $j$. In the base case $b_{i_1} \leq c_1$, as we have argued before. Thus consider the induction step. If $Y_j$ is non-empty, then by definition $Y_t$ is also non-empty for all $t<j$ and hence $c_t$ is defined. By induction $b_{i_t} \leq c_t$ for $1\leq t\leq j-1$. We note that for each $t < j$ the set of vertices adjacent to $v_{i_t}$, i.e, $X_t$ (if any) has been deleted from $Y_{t+1}$. Thus the vertices in $Y_j$ are not adjacent to the vertices $\{v_{i_1},\ldots,v_{i_{j-1}}\}$ chosen by MBF. Thus in iteration $j$ MBF can choose any vertex from $Y_j$ if it is left with sufficient budget. Now the budget used by MBF before iteration $j$ is $\sum_{t=1}^{j-1} b_{i_t}$. Thus the budget left is  $B -\sum_{t=1}^{j-1} b_{i_t} \geq B - \sum_{t=1}^{j-1} c_t \geq \sum_{v_l \in Y_j} b_l \geq \min_{v_l \in Y_j} b_l$, where the first inequality follows by induction. Thus the way MBF chooses a vertex it follows that $b_{i_j}$ should be lesser than or equal to the budget of the vertex in $Y_j$ with minimum budget. Thus $b_{i_j} \leq c_j$.
\end{proof}

Now we prove the following lemma which bounds the cardinality of $A^*$.

\begin{lemma}\label{lem:cardA*}
Suppose $A$ be the independent set returned by MBF and $A^*$ be an optimum solution, then $|A^*| \leq d|A|$. 
\end{lemma}

\begin{proof}
For each $1 \leq j\leq k$, $|A_j^*| \leq d$. As there are $k$ such sets $\sum_{j=1}^k |A_j^*| \leq dk=d|A|$. Thus it is sufficient to show that $A^* \setminus \cup_{j=1}^k A_j^*=\phi$. 

For the sake of contradiction suppose $A^* \setminus \cup_{j=1}^k A_j^*$ is non-empty. Let $v_t \in A^* \setminus \cup_{j=1}^k A_j^*$ such that $v_t$ has minimum budget $b_t$ among the vertices in $A^* \setminus \cup_{j=1}^k A_j^*$. Note that this means $Y_{j}\neq \phi$, for $1\leq j\leq k$. Now the budget used by MBF is $\sum_{j=1}^{k} b_{i_j}$. Thus the budget remaining is $B-\sum_{j=1}^{k} b_{i_j} \geq B-\sum_{j=1}^{k} c_j \geq b_t$, where the first inequality follows from Lemma \ref{lem:lessbudget}. Consider the set of vertices adjacent to $v_{i_j}$ in $A^*$ for $1 \leq j\leq k$, i.e, $\cup_{j=1}^k X_j$. By construction $\cup_{j=1}^k X_j \subseteq \cup_{j=1}^k A_j^*$. Thus the vertices in $A^* \setminus \cup_{j=1}^k A_j^*$ are not adjacent to the vertices $\{v_{i_1},\ldots,v_{i_{k}}\}$ chosen by MBF. As the remaining budget is at least $b_t$ after selecting the $k$ vertices, MBF could have chosen $v_t$ in its solution. Hence this is a contradiction and $A^* \setminus \cup_{j=1}^k A_j^*=\phi$, which completes the proof of this lemma.
\end{proof}

Lemma \ref{lem:cardA*} shows that MBF is a $d$-factor approximation algorithm for MBIS in any $d+1$-claw free graph. Note that MBF processes the vertices in non-decreasing order of their budgets. Thus to get this ordering we need to sort the vertices which takes $O(n\log n)$ time. In iteration $i$ MBF checks if $v_i$ is adjacent to the vertices already selected. Thus each edge is accessed at most once over all iterations. This takes $O(|E|)$ time if the graph is implemented using adjacency list data structure. Thus in total MBF runs in $O(|V|\log |V|+|E|)$ time. Hence we obtain the following theorem.

\begin{theorem}\label{th:clawfree}
There is a $d$-factor approximation algorithm for MBIS in any $d+1$-claw free graph $G=(V,E)$, that runs in $O(|V|\log |V|+|E|)$ time.
\end{theorem}

\textit{Remarks}. (1) The approximation ratio of MBF is tight indeed. Consider a $k_{1,d}$ graph whose center vertex has budget 1, each leaf has budget 2 and $B=2d$. In this case the maximum $B$-budgeted independent set has the cardinality $d$ while MBF returns the set containing only the center vertex.

(2) One might be tempted to extend MBF for weighted $d+1$-claw free graphs. A natural extension is to select the vertex with maximum weight-budget ratio (call this algorithm Maximum Weight-Budget Ratio First or MWBRF). But this heuristic doesn't work at all. We construct a family of $d+1$-claw free graphs such that for any natural number $M$ there is a graph $G$ in the family such that the solution of MWBRF on $G$ is as small as $\frac{1}{M}$ times of the optimum solution. For any $M$ we take a $K_{1,d}$ such that both the weight and budget of the central vertex is 1. The weight and budget of each leaf are $\frac{M}{d}$ and $M$. Also let the total budget $B$ is equal to $dM$. The weight-budget ratio of central vertex and any leaf are 1 and $\frac{1}{d}$. Thus MWBRF outputs the central vertex as its solution. The optimal is the set of all leaves. Thus the solution of MWBRF is $1 = \frac{1}{M}.M$, i.e, $\frac{1}{M}$ times of the optimum solution.

(3) The main reason that MBF gives a $d$-factor approximation in $d+1$-claw free graphs is that a vertex in the solution returned by MBF can be adjacent to at most $d$ vertices of the maximum solution. Moreover, any graph for which a vertex in the solution returned by MBF possess this bounded degree property, MBF yields an approximation factor equal to that bound. Thus in a bounded degree graph with maximum degree $\Delta$ MBF gives a $\Delta$-factor approximation and we have the following observation.
%

\section{MWBIS in Planar Graphs}\label{sec:planar}
In the early 80's Baker gave a PTAS for the MIS problem in any planar graph using a nontrivial dynamic programming based approach \cite{Baker83}. We adapt Baker's technique to get a PTAS for the MWBIS problem in any planar graph. Our approach is almost similar to Baker's technique except the part where the dynamic programming is extended to handle the budgets of the vertices. But for the sake of completeness we discuss the whole technique. Let $G$ be any given planar graph. The overall idea is as follows. We delete some vertices of $G$ to get a subgraph $G'$ consists of disjoint connected components such that no vertex of any component is adjacent to a vertex of any other component. Thus MWBIS can be solved independently in each of those components for any budget $0\leq B'\leq B$. To be precise in each of those components MWBIS can be solved exactly in linear time using a dynamic programming based approach. Once the solution in each component is known, the exact solution of MWBIS in $G'$ can be found by using the ALLOC routine of Section \ref{sec:tree}. We use a simple layering technique to construct the subgraph $G'$ such that the maximum solution of MWBIS in $G'$ is at least $\frac{k}{k+1}$ fraction of the maximum solution of MWBIS in $G$ for any fixed $k > 0$. Thus we get a $\frac{k}{k+1}$-approximation of MWBIS in $G$. Next we discuss each step in detailed manner. But before moving on we have some definitions.

We assume that we are given a planar embedding of $G$. All the edges and vertices in the unbounded face of this embedding forms a boundary of $G$. We define the level of each vertex of $G$ with respect to this boundary. The level of the vertices on this boundary is 1. The level of a vertex is $i$ ($> 1$) if it appears on the unbounded face of the graph obtained by deletion of level $j$ vertices from $G$ for all $1\leq j\leq i-1$. A planar graph is called $k$-outerplanar if it has a planar embedding where the level of any vertex is at most $k$. With all these definitions we start the discussion of our approach to solve MWBIS in $G$.

\subsection{Layering Technique}
Fix a positive integer $k$. For each $0 \leq i\leq k$, let $G_i$ be the graph obtained by deleting all the level $j$ vertices from $G$, where $j$ is congruent to $i$ mod $k+1$. Note that $G_i$ is a subgraph of $G$ consists of connected components, such that no vertex in any component is adjacent to a vertex of any other component. Moreover, each connected component is a $k$-outerplanar graph. Let $W_{MAX}$ be the weight of any maximum solution of MWBIS in $G$. Now by pigeon hole principle there is an index $0\leq r\leq k$ such that at most $\frac{1}{k+1}$ fraction of the weight $W_{MAX}$ comes from the level $j$ vertices, where $j$ is congruent to $r$ mod $k+1$. Thus we have the following observation.

\begin{obs}\label{obs:maxplanar}
There exists an index $r$ such that weight of any maximum solution of MWBIS in $G_r$ is at least $\frac{k}{k+1}W_{MAX}$, where $0\leq r\leq k$. 
\end{obs}

Suppose MWBIS can be solved in any $G_i$ exactly for $0\leq i\leq k$. We return the maximum weight among these $k+1$ solutions. Then by Observation \ref{obs:maxplanar} we have a $\frac{k}{k+1}$-factor approximation for MWBIS in $G$. Now for any $\epsilon > 0$, setting $k=\ceil{\frac{1}{\epsilon}}$ gives an $(1-\epsilon)$-factor approximation for MWBIS in any planar graph. Now consider the MWBIS problem in any $G_i$. Note that each $G_i$ is a collection of $k$-outerplanar graphs. Later we give a proof sketch of the following theorem.

\begin{theorem}\label{th:k-outerplanar}
For any $k > 0$, MWBIS can be solved in any $k$-outerplanar graph $G'=(V',E')$ in $O(8^k|V'|B^3)$ time.
\end{theorem}

Assuming Theorem \ref{th:k-outerplanar} is true we show how the solutions in the components of $G_i$ can be merged to get an exact solution for $G_i$. 

\subsection{Merging the Solutions}
Consider the graph $G_i$ for a fixed $i$. Let $G_i$ consists of the set $\{G_{i1},\ldots,G_{il}\}$ of $l$ $k$-outerplanar graphs. We solve the MWBIS problem in each graph $G_{it}$ for budget $B'$, where $1\leq B'\leq B$. Then we make a call to ALLOC($f_1$, $\ldots$,$f_l$;$p$) with $f_t(p_t)=OPT(G_{it},p_t)$ for $1\leq t\leq l$ and $p=B$. ALLOC returns the maximum solution of MWBIS in $G_i$. Thus for all $0\leq i\leq k$ the $k+1$ calls to ALLOC takes $O(knB^2)$ time. Thus we get the following corollary from Theorem \ref{th:k-outerplanar}.

\begin{cor}
There is a PTAS for MWBIS in planar graphs that runs in $O(8^{O(\frac{1}{\epsilon})}nB^3)$ time.
\end{cor}

\subsection{MWBIS in k-Outerplanar Graphs}
To focus on the dynamic programming technique we consider the MWBIS problem in a special class of $k$-outerplanar graphs. The technique can be adapted to solve MWBIS in any general $k$-outerplanar graph in a similar way as Baker did. Consider a planar embedding of any $k$-outerplanar graph $G$. A level $i$ face in $G$ is a cycle consists of level $i$ vertices. $G$ is called \textit{simply nested} if it contains exactly one level $i$ face for all $i$. In the remaining of this subsection we consider MWBIS in any given simply nested $k$-outerplanar graph $G$. We assume that a planar embedding of $G$ is given.

The overall idea is as follows. We divide $G$ into ``slices'' where each ``slice'' is a subgraph of $G$ containing at least one vertex of each level. We solve MWBIS in each of these slices using dynamic programming. For any such ``slice'' at first MWBIS is solved in its induced subgraph containing edges and vertices of level at most $i$. Then this solution is used to get a solution for the induced subgraph of the ``slice'' containing edges and vertices of level at most $i+1$. At last the solution in all ``slices'' are merged to get a solution for the whole graph. 

We define a level $i$ slice corresponding to an edge on any level $i$ face in an inductive way. Each level $i$ slice has two boundaries that separate it from other level $i$ slices. To be precise the slices corresponding to any two consecutive edges on the level $i$ face share a common boundary. For any edge on the level 1 face the boundaries of the slice are the end vertices. The slice contains the edge itself and its vertices. Now say the level $i-1$ slices are already defined for $i > 1$. Let $e_1,\ldots,e_p$ be the edges on the level $i$ face in counterclockwise order. Suppose $e_1=(u,v)$, where $v$ is an endpoint of $e_2$ as well. Let $u'$ be a level $i-1$ vertex such that a line segment can be drawn between $u$ and $u'$ without crossing any edge (see Figure \ref{fig:fig2}(a)). Also let $e_1',\ldots,e_l'$ be the edges on the level $i-1$ face in counterclockwise order, where $u'$ is an endpoint of $e_1'$ and $e_l'$. The first slice boundary of $e_1$ and the second slice boundary of $e_p$ are same, and consists of $u$ and the vertices in the first slice boundary of $e_1'$. The boundary between the slices of $e_j$ and $e_{j+1}$ is defined in the following way for $1\leq j\leq p-1$. Let $x$ be the common end vertex of $e_j$ and $e_{j+1}$, and $y$ be a level $i-1$ vertex such that a line segment can be drawn between $x$ and $y$ without crossing any edge. Suppose $m$ be the index such that $y$ is a common vertex of $e_m'$ and $e_{m+1}'$. Then the slice boundary between $e_j$ and $e_{j+1}$ consists of $x$ and the vertices in second slice boundary of $e_m'$. A slice for an edge contains the edge itself, its end vertices, the boundary vertices, any edge between the boundary vertices, and all the edges and vertices between the boundaries. See Figure \ref{fig:fig2}(b) for an illustration.

\begin{figure*} 
\hspace{.14in}
  \begin{minipage}[c]{0.5\textwidth}
  \centering
  \includegraphics[width=60mm] 
    {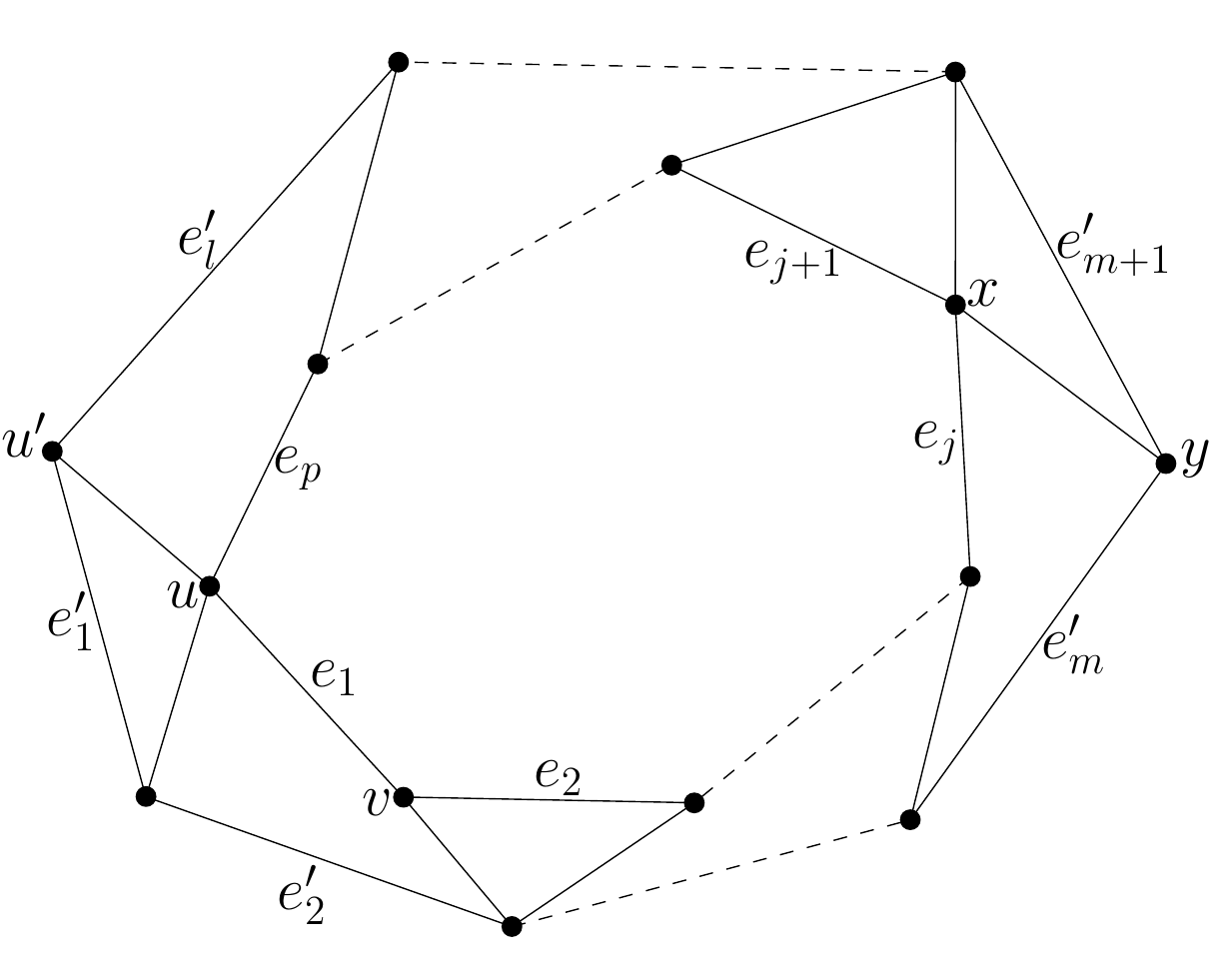}\\
    {\small (a)}\\
    \end{minipage}%
 \hspace{-.1in}
  \begin{minipage}[c]{0.5\textwidth}
  \centering
  \includegraphics[width=60mm]
    {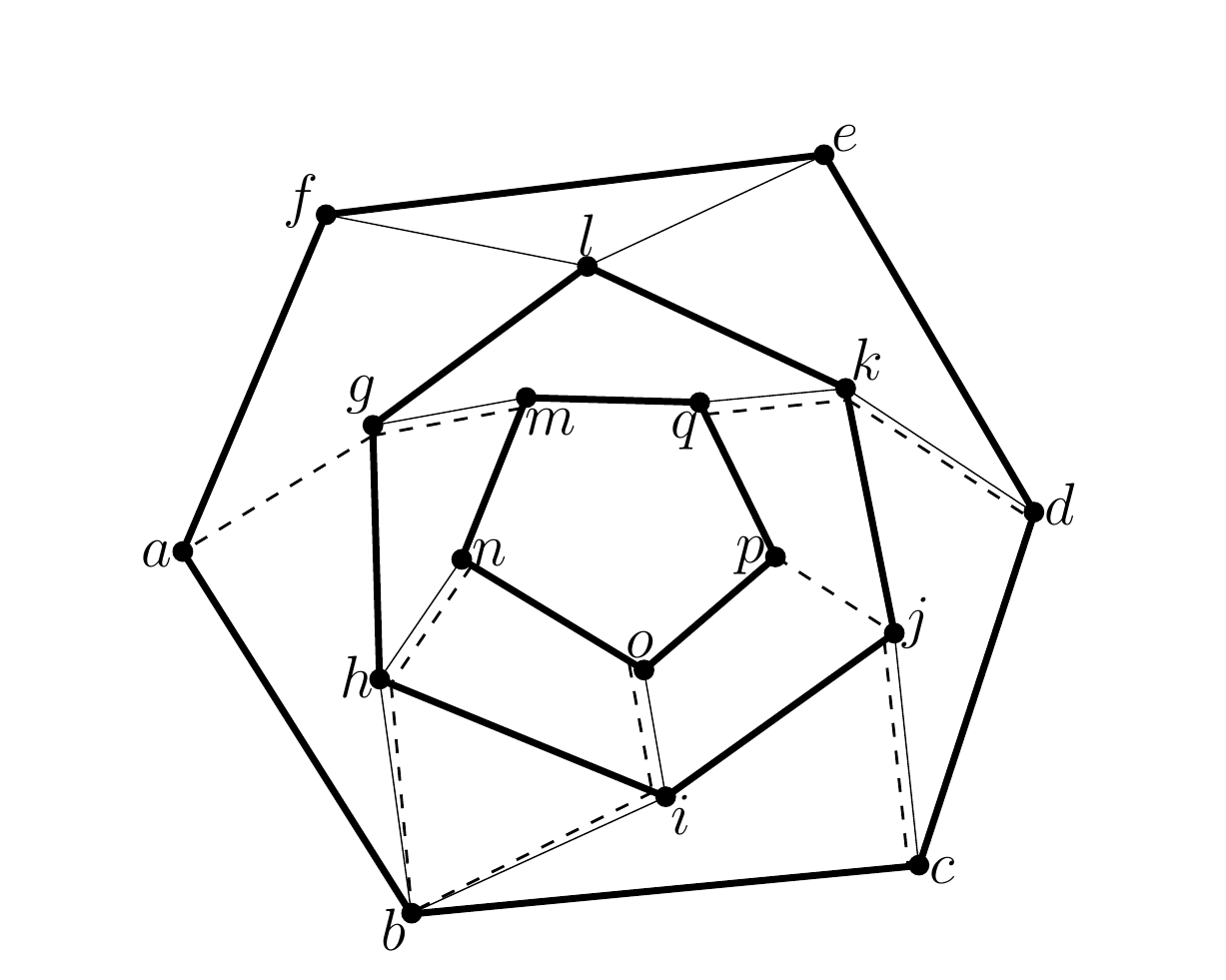}\\
    {\small (b)}\\
    \end{minipage}%
   \hspace{-3.6in}
  \caption{\textit{(a) Construction of slice boundaries. (b) Example of level 3 slices in a 3-outerplanar graph. Slice boundaries are shown using dashed lines. The face edges are shown using bold lines.}}
  \label{fig:fig2}  
\end{figure*}

\textit{Dynamic Programming}. For any level $i$ slice we maintain a table indexed by a length $i+1$ vector. The first $i$ values in this vector are binary tuples $(a_j,b_j)$ for $1\leq j\leq i$. $a_j$ and $b_j$ are corresponding to the level $j$ vertices in the first and second boundary of the slice. Each such collection of $i$ binary tuples represents a combination of the boundary vertices; 1 denotes the vertex is in the combination. The last value in the vector is an integer between $0$ and $B$. Thus we denote an index by the 3-tuple $(\bar{a},\bar{b},B')$, where $\bar{a},\bar{b}$ are two bit vectors corresponding to first and second boundary, and $B'$ is an integer. For a table $T$, we denote the entry corresponding to the index $(\bar{a},\bar{b},B')$ by $T[\bar{a},\bar{b},B']$. $T[\bar{a},\bar{b},B']$ stores the weight of any optimum solution of MWBIS in the slice with budget at most $B'$, such that the solution includes all the vertices in the combination. The value corresponding to an index is invalid if the combination of vertices is not an independent set, or the budget of the combination exceeds $B'$. Thus a level $i$ table stores $4^{i}B$ entries. A table for a level $i$ ($i > 1$) slice is computed by merging relevant tables corresponding to level $i-1$ slices. The merging process is as follows.

\textit{Merging of Two Tables}. To merge two tables $T_1$ and $T_2$ corresponding to two level $i$ slices $s_1$ and $s_2$, the slices should be adjacent, i.e. second boundary of $s_1$ must be same as the first boundary of $s_2$. The new table $T$ is corresponding to the subgraph which is the union of the two slices. For each pair of length $i$ vectors $\bar{a}$ and $\bar{c}$, and an integer $1\leq B'\leq B$, we compute the entry $T[\bar{a},\bar{c},B']$. For each length $i$ vector $\bar{b}$, we make a call to ALLOC($f_1$,$f_2$;$p$), where $f_1(p_1)=T_1[\bar{a},\bar{b},p_1]$, $f_2(p_2)=T_2[\bar{b},\bar{c},p_2]$, and $p=B'$. As some boundary vertices appear in both slices we need to subtract their weights from the solution returned by ALLOC. The maximum solution over all the values of $\bar{b}$ is stored in $T[\bar{a},\bar{c},B']$.

\textit{The Algorithm: Construction of Tables for the Slices}. For any level 1 slice, the graph consists of only two vertices and the entries of the table are computed trivially. Now consider any level $i$ ($i > 1$) slice corresponding to an edge $(u,v)$. Let $u'$ and $v'$ be the level $i-1$ vertices in the first and second boundary of the slice respectively. If $u'=v'$, the table can be computed trivially. Otherwise, say $e_1,\ldots,e_m$ be the edges between $u'$ and $v'$ on the level $i-1$ face in counterclockwise order. Suppose the table corresponding to each $e_j$ is already computed for $1\leq j\leq m$. All the tables corresponding to these $m$ edges are merged together. First the table for $e_1$ and $e_2$ are merged. Then the new table is merged with the table of $e_3$ and so on. The final table $T'$ we get is corresponding to the subgraph which is the union of all the slices of $e_1,\ldots,e_m$. Note that the slice corresponding to $(u,v)$ is the union of this subgraph, the edge $(u,v)$, the vertices $u$, $v$ and possibly the edges $(u,u')$ and $(v,v')$ if any. Thus to extend $T'$ to a table $T$ for the level $i$ slice, we need to consider 4 more combinations depending on whether $u$ or $v$ is present in the solution. If any combination contains none of $u$ and $v$, the entries in $T$ remain same as the corresponding entries of $T'$. Consider any other valid index $(\bar{a},\bar{b},B')$, where $\bar{a}$ and $\bar{b}$ are length $i+1$ vectors, such that $u$ is chosen in $\bar{a}$, but $v$ is not chosen in $\bar{b}$. Then $T[\bar{a},\bar{b},B']=T'[\bar{a'},\bar{b'},B'-b_u]+w_u$, where $\bar{a'}$ and $\bar{b'}$ are the length $i$ vectors obtained by dropping the last scalar values from $\bar{a}$ and $\bar{b}$, and $w_u$ and $b_u$ are the weight and budget of $u$. Similarly, the entries can be computed for the combinations where $u$ is not chosen, but $v$ is chosen.

Let $l_m$ be the maximum among the levels of the vertices of $G$. Then by merging all the level $l_m$ slices, two at a step, we get a table $T_G$ for the whole graph. Recall that while merging two tables we compute an entry of the table for each index $\bar{a},\bar{c},B'$. In the last step of merging the vectors $\bar{a}$ and $\bar{c}$ must be same, as they are corresponding to the same boundary. Lastly, the solution of MWBIS in $G$ with budget $B$ can be found by taking the maximum among all the entries $T_G[\bar{a},\bar{c},B]$. 

Now we argue that any table stores the maximum solution corresponding to all possible combination of the boundary vertices. In base case the tables are created for the slices which are single edges and the correctness trivially follows. In inductive step the table for a level $i$ slice is computed by extending the tables for level $i-1$ slices considering all possible combinations of the boundary vertices. As the level $i-1$ tables contain correct entries by induction, the correctness for the entries of any level $i$ table also follows. 

The time complexity of the algorithm is dominated by the complexity of the table construction for the slices. For a single level $i$ slice the overall calls to ALLOC takes $O(8^kn_1B^3)$ time, where $n_1$ level $i-1$ slices are processed to form the new table. As the slices are edge disjoint in a single level except in the boundaries, each slice is used at most once for computation of the other slices. Thus in total the algorithm takes $O(8^knB^3)$ time, where $n$ is the number of vertices of $G$. As mentioned earlier, the same dynamic programming can be extended to handle any general $k$-outerplanar graph $G'=(V',E')$, which takes $O(8^k|V'|B^3)$ time.

\section{MWBIS in Interval Graphs}\label{sec:interval}
Given an interval graph $G=(V,E)$ we consider the MWBIS problem in $G$. We assume that the intervals corresponding to the vertices are also given. Let $I_1,\ldots,I_n$ be the intervals sorted in non-decreasing order of their finishing time. Suppose $b_j$ and $w_j$ be the respective budget and weight of $I_j$. An independent set of intervals is defined to be a set of pairwise non-intersecting intervals. We design a dynamic programming based algorithm to find a maximum weight independent set of these intervals having total budget at most $B$. Define $V_j$ to be the set of intervals $\{I_1,\ldots,I_j\}$
. 

\begin{algorithm*}[]
 \caption{}
\label{alg:interval}
\begin{algorithmic}[1]
   \REQUIRE A set of intervals $I=\{I_1,\ldots,I_n\}$ sorted in non-decreasing order of their starting time, a weight function $w$, a budget function $b$, a positive integer $B$, a $(n+1) \times (B+1)$ table $M$ 
   \ENSURE a maximum $B$-budgeted independent set of $I$ and its weight
   \FOR {$t=0$ to $B$}
     \STATE $w_{0,t} \leftarrow 0$
   \ENDFOR
   \FOR {$j=1$ to $n$}
     \STATE $w_{j,0} \leftarrow 0$
   \ENDFOR
   \FOR {$j=1$ to $n$}
     \FOR {$t=1$ to $B$}
       \STATE $w_{j,t} \leftarrow w_{j-1,t}$
       \STATE $f_{j,t} \leftarrow 0$
       \IF {$t \geq w_j$}  
         \STATE set $l_j$ to be the maximum index $m$ such that $I_{m} \in V_j$ does not intersect $I_j$ \COMMENT{set $l_j$ to 0 if no such maximum index exists}
         \IF {$w_{j,t} < w_{l_j,t-b_j}+w_j$}
           \STATE $w_{j,t} \leftarrow w_{l_j,t-b_j}+w_j$
           \STATE $f_{j,t} \leftarrow 1$
         \ENDIF
       \ENDIF
     \ENDFOR
   \ENDFOR
   \STATE $I' \leftarrow \phi$
   \STATE $j \leftarrow n$
   \STATE $t \leftarrow B$
   \WHILE {$t > 0$}
     \IF {$f_{j,t}==1$}
       \STATE $I' \leftarrow I' \cup {I_j}$
       \STATE $j \leftarrow l_j$
       \STATE $t \leftarrow t-b_j$
     \ELSE
       \STATE $j \leftarrow j-1$
     \ENDIF 
   \ENDWHILE  
   \RETURN $I', w_{n,B}$
\end{algorithmic}
\end{algorithm*}
Our algorithm processes the intervals in the order mentioned before. In iteration $j$ it computes maximum weight independent sets of $V_j$ having budget $k$ for $1\leq k \leq B$. Consider a maximum weight independent set of $V_j$ having budget $k$. Note that there could be two cases: (i) $I_j$ is contained in it, and (ii) $I_j$ is not contained in it. In Case (i) the solution is composed of $I_j$, and a maximum weight independent set of the intervals of $V_j$ which doesn't intersect $I_j$ and having budget at most $k-b_j$. Let $l_j$ be the maximum index such that $I_{l_j} \in V_j$ doesn't intersect $I_j$. As the intervals are sorted in nondecreasing order of their finishing time, no interval $I_t$ intersects $I_j$ for $t < l_j$. Thus if a maximum weight independent set of $V_{l_j}$ having budget $k-b_j$ is already computed, then we can use that solution to compute a solution for $V_j$. In Case (ii) the solution is a maximum weight independent set of $V_{j-1}$ having budget $k$. Thus in this case also we can readily compute a solution for $V_j$ if the solution for $V_{j-1}$ is already computed.

As the values computed in some iteration might be needed in later iterations all the computed values are stored in a 2-dimensional table $M$. Each cell $M(j,t)$ of $M$ contains two values $w_{j,t}$ and $f_{j,t}$, where $1\leq j\leq n$, $0\leq t\leq B$. $w_{j,t}$ stores the weight of maximum weight independent set of $V_j$ having budget $t$ and $f_{j,t}$ is the flag indicating whether $I_j$ is in that maximum weight independent set ($f_{j,t}=1$) or not ($f_{j,t}=0$). Now we formally describe the algorithm that computes a maximum weight independent set of $V_n$ of budget $B$. 

The entries of the row 0 and column 0 of $M$ are initialized to 0. The entries of the row 0 are used to solve the subproblems for which the input set of intervals is empty. The algorithm then considers the two cases as mentioned above and fills up all the entries of $M$. The maximum weight independent set $I'$ is retrieved by traversing in the backward direction starting with $I_n$. The flag value $f_{j,t}$ is consulted to decide whether $I_j$ is chosen in the solution. Lastly, the entry $w_{n,B}$ stores the weight of the compuetd set $I'$. Thus the algorithm returns its value with $I'$.

The time complexity of this algorithm is dominated by the computation of $l_j$ values and complexity of the nested for loop. Note that all the $l_j$ values can be precomputed outside of the nested for loop in quadratic time. Thus the nested for loop takes $O(nB)$ time and we have the following theorem.

\begin{theorem}
In any interval graph the MWBIS problem can be solved in $O(n(n+B))$ time.
\end{theorem}

\section{Discussion} 
All the exact algorithms we have designed in this paper has a running time dependent on $B$ which could be exponential to $n$. In fact one can prove that MWBIS is $\mathcal{NP}$-hard even in star trees by reducing the problem from KNAPSACK \cite{GareyJ79}. Given an instance $I$ of KNAPSACK having $n$ objects, we construct a tree $K_{1,n}$. Each leaf in $K_{1,n}$ is corresponding to an object having same weight and budget (size) as the object. 
The budget of the center vertex is set to $B+1$, where $B$ is the size constraint in $I$. Then the budget is set to $B$ and thus the center vertex is never selected in any solution of MWBIS. Now it is easy to see that for the instance $I$, KNAPSACK has a solution of value at least $k$ if and only if there is a $B$-budgeted independent set of weight at least $k$ in the constructed tree. 
The problem that still remains unsolved is to determine the classes of graphs for which it is possible to design $O(n^c)$ time algorithms for some constant $c$.

Is it possible to improve the approximation bound in case of MBIS in $d+1$-claw free graphs? What can we say about the weighted version of this problem?
For planar graphs we have designed a PTAS for MWBIS using Baker's technique. The same technique can be extended to solve MWIS in bounded treewidth graphs \cite{Bodlaender93}. Hence an interesting problem in this direction is to extend our dynamic programming strategy for the bounded treewidth graphs.

In case of geometric graphs (intersection graphs of geometric objects) the hardness of approximation results for MIS do not hold. In fact there is a vast literature of MIS that propose PTASs and QPTASs for interesting classes of geometric graphs (\cite{AdamaszekW14},\cite{Chan03},\cite{ChanH12},\cite{ErlebachJS05},\cite{HochbaumM85}). Thus another possible research direction is to study the MWBIS problem in geometric graphs. 
\\\\\textbf{Acknowledgements}\\
I am deeply thankful to Kasturi Varadarajan for pointing out the fact that Baker's technique can be extended in case of MWBIS in planar graphs.

\bibliographystyle{plain}
\bibliography{employee}
\end{document}